\documentclass[oribibl]{llncs}

\usepackage{amssymb}
\usepackage{amsmath}   %
\usepackage{amsfonts}
\usepackage[bookmarksnumbered,colorlinks,linkcolor=blue,citecolor=blue,urlcolor=blue,
  pdftitle={Quantum algorithm for the Boolean hidden shift problem},
  pdfauthor={Dmitry Gavinsky, Martin Roetteler and Jeremie Roland}
]{hyperref}

\date{}
\newsavebox{\fmbox}
\newenvironment{fmpage}[1]
     {\medskip\begin{lrbox}{\fmbox}\begin{minipage}{#1}}
     {\end{minipage}\end{lrbox}\fbox{\usebox{\fmbox}}\medskip}
\newcommand{\algorithm}[1]{
\begin{center}
\begin{fmpage}{\textwidth}
#1
\end{fmpage}
\end{center}}

\newtheorem{lemma2}{Lemma}

\newcommand{\ket}[1]{|#1\rangle}

\newcommand{\sz}[1]{\left|#1\right|}
\newcommand{\abs}[1]{\left|#1\right|}

\newcommand{\wh}{\widehat}
\newcommand{\wt}{\widetilde}

\newcommand{\gv}{\gamma_{f,v}}
\newcommand{\gf}{\gamma_{f}}
\newcommand{\tf}{\wt{F_v}}

\newcommand{\DU}{D_f^U}
\newcommand{\qc}{t_{\mathrm{cla}}}

\newcommand{\Ac}{\mathcal{A}_{\mathrm{cla}}}

\newcommand{\BHSP}{BHSP}
\newcommand{\RR}{\mathbb{R}}
\newcommand{\ZZ}{\mathbb{Z}}

\newcommand{\Span}{\mathsf{Span}}
\newcommand{\deq}{\stackrel{\textrm{def}}{=}}
\newcommand{\ip}[2]{\langle #1,#2\rangle}
\newcommand{\PR}[2][]{\mathop{\mathbf{Pr}}_{#1}{\left[{#2}\right]}}
\newcommand{\E}[2][]{\mathop{\mathbf{E}}_{#1}{\left[{#2}\right]}}

\newcommand{\MyComment}[1]{\ClassWarning{My Macros}{#1}}
\newcommand{\mynb}[1]{\MyComment{A valuable comment may follow...}}

\title{Quantum algorithm for the Boolean\\ hidden shift problem}

\author{Dmitry Gavinsky\thanks{$\{$dmitry,mroetteler,jroland$\}$@nec-labs.com}
 \and Martin Roetteler$^\star$
 \and J{\'e}r{\'e}mie Roland$^\star$
}
\institute{NEC Laboratories America, Inc.}

\begin{document}

\maketitle

\thispagestyle{empty}

\begin{abstract}
The hidden shift problem is a natural place to look for new separations between classical and quantum models of computation. One advantage of this problem is its flexibility, since it can be defined for a whole range of functions and a whole range of underlying groups. In a way, this distinguishes it from the hidden subgroup problem where more stringent requirements about the existence of a periodic subgroup have to be made. And yet, the hidden shift problem proves to be rich enough to capture interesting features of problems of algebraic, geometric, and combinatorial flavor. We present a quantum algorithm to identify the hidden shift for any Boolean function. Using Fourier analysis for Boolean functions we relate the time and query complexity of the algorithm to an intrinsic property of the function, namely its minimum influence. We show that for randomly chosen functions the time complexity of the algorithm is polynomial. Based on this we show an average case exponential separation between classical and quantum time complexity. A perhaps interesting aspect of this work is that, while the extremal case of the Boolean hidden shift problem over so-called bent functions can be reduced to a hidden {\em subgroup} problem over an abelian group, the more general case studied here does not seem to allow such a reduction.
\end{abstract}

\setcounter{page}{1}

\section{Introduction}

Hidden shift problems have been studied in quantum computing as they
provide a framework that can give rise to new quantum algorithms. The
hidden shift problem was first introduced and studied in a paper by
van Dam, Hallgren and Ip \cite{vDHI:2003} and is defined as follows.
We are given two functions $f$, $g$ that map a finite group $G$ to
some set with the additional promise that there exists an element
$s\in G$, the so-called shift, such that for all $x$ it holds that
$g(x) = f(x+s)$. The task is to find $s$. Here the group $G$ is
additively denoted, but the problem can be defined for non-abelian
groups as well. The great flexibility in the definition allows to
capture interesting problems ranging from algebraic problems such as
the shifted Legendre symbol \cite{vDHI:2003}, over geometric problems
such as finding the center of shifted spheres \cite{CSV:2007,Liu:2009}
and shifted lattices \cite{Regev:2004}, to combinatorial problems such
as graph isomorphism \cite{CW:2007}.

Notable here is a well-known connection between the hidden subgroup
problem for the dihedral group, a notoriously difficult instance which
itself has connections to lattice problems and average case subset sum
\cite{Regev:2004} and a hidden shift problem over the cyclic group
$\ZZ_n$ where the functions $f$ and $g$ are injective
\cite{Kuperberg:2005,MRRS:2007,CvD:2007}. It is known \cite{FIMSS:2003,Kuperberg:2005} that the hidden
shift problem for {\em injective} functions $f,g: G \rightarrow S$
that map from an abelian $G$ to a set $S$ is equivalent to hidden subgroup
problem over the semi-direct product between $G$ and $\ZZ_2$, where the action of $\ZZ_2$ on $G$ is
given by the inverse. We would like to point out that the functions studied here are Boolean functions (i.e., $G=\ZZ_2^n$) and therefore
far from being injective. Even turning them into injective {\em quantum}
functions, as is possible for bent functions \cite{Roetteler:2010},
seems not to be obvious in this case. 
 Another recent example of a
non-abelian hidden shift problem arises in a reduction used to argue
that the McEliece cryptosystems withstands certain types of quantum
attacks \cite{DMS:2010}.

In this paper we confine ourselves to the abelian case and in
particular to the case where $G=\ZZ_2^n$ is the Boolean hypercube. The
resulting hidden shift problem for Boolean functions, i.e., functions
that take $n$ bits as inputs and output just $1$ bit, at first glance
looks rather innocent. However, to our knowledge, the Boolean case was
previously only addressed for two extreme cases: a) functions which
mark precisely one element and b) functions which are maximally apart
from any affine Boolean function (so-called bent functions). In case
a), the problem of finding the shift is the same as unstructured
search, so that the hidden shift can be found by Grover's algorithm
\cite{Grover:96} and the query complexity is known to be tight and is
given by $\Theta(\sqrt{2^n})$.

In case b) the hidden shift can be discovered in one query using an
algorithm that was found by one of the co-authors
\cite{Roetteler:2010}, provided that the {\em dual} of the function
can be computed efficiently, where the definition of the dual is via
the Fourier spectrum of the function which in this case can be shown
to be flat in absolute value.  If no efficient implementation of the
dual is known then still a quantum algorithm exists that can identify
the hidden shift in $O(n)$ queries. The present paper can be thought
of as a generalization of this latter algorithm to the case of Boolean
functions other than those having a flat spectrum. This is motivated
by the quite natural question of what happens when the extremal
conditions leading to the family of bent functions are relaxed. In this paper we
address the question of whether there is a broader class of functions for
which hidden shifts of a function can be identified.

The first obvious step in direction of a generalization is actually a
roadblock: Grover's search problem~\cite{Grover:96} can also be cast
as a hidden shift problem. In this case the corresponding class of
Boolean functions are the \emph{delta functions}, i.e., $f, g:
\{0,1\}^n \rightarrow \{0,1\}$, where $g(x)=f(x+s)$ and $f(x)$ is the
function that takes value $1$ on input $(0, \ldots, 0)$ and $0$
elsewhere and $g(x)$ is the function that takes the value $1$ on input
$s$ and $0$ elsewhere.  Grover's algorithm \cite{Grover:96} allows to
find $s$ in time $O(\sqrt{2^n})$ on a quantum computer (which is also
the fastest possible \cite{BV:97}).

Thus, the following situation emerges for the quantum and the classical query complexities of these two extremal cases: for bent functions the classical query complexity\footnote{Note that the query complexity depends crucially on how the functions $f$ and $g$ can be accessed: the stated bounds hold for the case where $f$ and $g$ are given as black-boxes. If $f$ is a {\em known} bent function, then it is easy to see that the classical query complexity becomes $O(n)$.} is $\Omega(\sqrt{2^n})$ and the quantum query complexity\footnote{A further improvement is possible in case the so-called {\em dual} bent function
  $\widetilde{f}$ is accessible via another black-box: in this case
  the quantum query complexity becomes constant
  \cite{Roetteler:2010}.} is $O(n)$.  For delta functions the classical query
complexity is $\Theta(2^n)$ and the quantum query complexity is
$\Theta(\sqrt{2^n})$.

For a general Boolean function the hidden shift problem can be seen as
lying somewhere between these two extreme cases. This is somewhat
similar to how the so-called weighing matrix
problem~\cite{WeighingMatrices} interpolates between the
Bernstein-Vazirani problem~\cite{BV:97} and Grover search, and how the
generalized hidden shift problem \cite{CvD:2007} interpolates between
the abelian and dihedral hidden subgroup problems. However, apart from
these two extremes, not much is known about the query complexity of
the hidden shift problem for general Boolean functions.

The main goal of this work was to understand the space between these
two extremes. We show that there is a natural way to ``interpolate''
between them and to give an algorithm for each Boolean function whose
query complexity depends only on properties of the Fourier spectrum of
that function.

\bigskip

\paragraph{\bf Prior work.} As far as hidden shifts of Boolean
functions are concerned, besides the mentioned papers about the bent
case and the case of search, very little was known. The main technique
previously used to tackle hidden shift problem was by computing a
suitable convolution. However, in order to maintain unitarity, much of
target function's features that we want to compute the convolution
with had to be ``sacrificed'' by requiring the function to become
diagonal unitary, leading to a renormalization of the diagonal
elements, an issue perhaps first pointed out by \cite{CurtisMeyer}.
No such renormalization is necessary if the spectrum is already flat
which corresponds to the case of the Legendre symbol \cite{vDHI:2003}
(with the exception of one special value at 0) and the case of
bent functions which was considered in \cite{Roetteler:2010}.

\paragraph{\bf Our results.} We introduce a quantum algorithm that allows
us to sample from vectors that are perpendicular to the hidden shift
$v$ according to a distribution that is related to the Fourier
spectrum of the given Boolean function $f$. If $f$ is bent, then this
distribution is uniform which in turn leads to a unique
characterization of $v$ from $O(n)$ queries via a system of linear
equations. For general $f$ more queries might be necessary and
intuitively the more concentrated the Fourier spectrum of $f$ is, the
more queries have to be made: in the extreme case of a ($\pm 1$
valued) delta function $f$ the spectrum is extremely imbalanced and
concentrated almost entirely on the zero Fourier coefficient which
corresponds to the case of unstructured search for which our algorithm
offers no advantage over Grover's algorithm. For general $f$ we give
an upper bound on the number of queries in terms of the {\em
  influence} $\gf$ of the function $f$, where the influence is defined
as $\gf = \min_{v} (\PR[x]{f(x)\neq f(x+v)})$.

From a simple application of the Chernoff bound it follows that it is
extremely unlikely that a randomly chosen Boolean function will give
rise to a hard instance for our quantum algorithm. This in turn gives 
rise to our main result of the paper: 

\bigskip
\noindent
{\bf Theorem 2 (Average case exponential separation).} {\em Let $({\cal O}_f,
  {\cal O}_g)$ be an instance of a Boolean hidden shift problem (BHSP)
  where $g(x)=f(x+v)$ and $f$ and $v$ are chosen uniformly at random.
  Then there exists a quantum algorithm which finds $v$ with bounded error using $O(n)$
  queries and in $O(\mathrm{poly}(n))$ time whereas any classical algorithm needs
  $\Omega(2^{n/2})$ queries to achieve the same task.}
 \smallskip

This result can be interpreted as an exponential quantum-classical
separation for the time and query complexity of an average case problem. Finally,
we would like to comment on the relationship between the problem
considered in this paper and the abelian hidden subgroup problem.  It
is interesting to note, yet not particularly difficult to see, that the
case of a hidden shift problem for bent functions can be reduced to
that of an abelian hidden subgroup problem. The hiding function in 
this case is a quantum function, i.\,e., it takes values in the set 
of quantum sets rather than just basis states. For the case of a 
non-bent function, including the cases of random functions considered 
here, the same direct correspondence to the hidden subgroup problem over 
an abelian group no longer exists, i.\,e., even though there is no obvious 
group/subgroup structure present in the function $f$, the algorithm 
can still identify the hidden shift $v$. 

\section{Preliminaries}

\begin{definition}[Boolean Hidden Shift Problem]\label{d_BHSP}
 Let $n \geq 1$ and let $f,g : \ZZ_2^n\to\ZZ_2$ be two Boolean functions such that the following conditions hold:
\begin{itemize}
 \item if for some $t\in \ZZ_2^n$ it holds that $f(x)\equiv f(x+t)$ then $t=0$;
 \item for some $s\in \ZZ_2^n$ it holds that $g(x)\equiv f(x+s)$.
\end{itemize}
If $f$ and $g$ are given by two oracles $O_f$ and $O_g$, we say that
the pair $(O_f,O_g)$ defines an instance of a hidden shift problem
(\BHSP) for the function $f$.  The value $s\in\ZZ_2^n$ that satisfies
$g(x)\equiv f(x+s)$ is the solution of the given instance of the
\BHSP.
\end{definition}

We also consider the $\{+1,-1\}$-valued function $F$ corresponding to
the function $f$ and view it as a function over $\RR$, that is,
\begin{align}
 F:\ZZ_2^n\to\RR:x\mapsto (-1)^{f(x)}.
\end{align}
The arguments of these functions are assumed to belong to $\ZZ_2^n$,
and their \emph{inner product} is defined accordingly, i.e., $
\ip{u}{v}=\bigoplus_{i=1}^n u_i\cdot v_i.  $ We also denote by
$\chi_u(\dots)$ the elements of the standard Fourier basis
corresponding to $\ZZ_2^n$, that is, $\chi_u(v)=(-1)^{\ip uv}$ for
every $u,v\in\ZZ_2^n$.


We will see that the complexity of the \BHSP\ depends on the notion of
influence.
\begin{definition}[Influence]
  For any Boolean function $f$ over $\ZZ_2^n$ and $n$-bit string $v$,
  we call $\gv=\PR[x]{f(x)\neq f(x+v)}$ the influence of $v$ over $f$,
  and $\gf=\min_{v}\gv$ the minimum influence over $f$.
\end{definition}

The following lemma relates the influence over a Boolean function $f$
to the Fourier spectrum of its $\{+1,-1\}$-valued analog $F$, see 
also \cite[Fact 11, p. 14]{Gopalan2009}. 

\begin{lemma}\label{lem:influence}
 $\gv=\sum_{u:\ip vu=1} \abs{\wh F(u)}^2.$
\end{lemma}
We give a proof of this lemma in Appendix \ref{app:influence} for completeness.

\section{Our algorithm}
\begin{theorem}\label{thm:complexity}
  There exists a quantum algorithm that solves an instance of \BHSP\
  defined over the function $f$ using expected $O(n/\sqrt{\gf})$
  oracle queries.  The algorithm takes expected time polynomial in the
  number of queries.
\end{theorem}

\begin{figure}
\centerline{
\unitlength0.75pt%
\begin{picture}(10,50)(0,0)
\put(-25,0){\makebox(0,0)[l]{$\ket{0}^{\otimes n}$}}
\put(-25,40){\makebox(0,0)[l]{$\ket{0}$}}
\end{picture}%
\begin{picture}(16,50)(0,0)
\put(0,0){\line(1,0){16}}
\put(6,-7){\line(1,3){5}}
\put(0,40){\line(1,0){16}}
\end{picture}%
\begin{picture}(30,50)(0,0)
\put(0,-15){\framebox(30,30){$H^{\otimes n}$}}
\put(0,40){\line(1,0){30}}
\end{picture}%
\begin{picture}(10,50)(0,0)
\multiput(0,0)(0,40){2}{\line(1,0){10}}
\end{picture}%
\begin{picture}(30,50)(0,0)
\put(0,-15){\framebox(30,70){$O_f$}}
\end{picture}%
\begin{picture}(10,50)(0,0)
\put(0,0){\line(1,0){10}}
\put(0,40){\line(1,0){10}}
\end{picture}%
\begin{picture}(30,50)(0,0)
\put(0,25){\framebox(30,30){$Z$}}
\put(0,0){\line(1,0){30}}
\end{picture}%
\begin{picture}(10,50)(0,0)
\put(0,0){\line(1,0){10}}
\put(0,40){\line(1,0){10}}
\end{picture}%
\begin{picture}(30,50)(0,0)
\put(0,-15){\framebox(30,70){$O_g$}}
\end{picture}%
\begin{picture}(10,50)(0,0)
\multiput(0,0)(0,40){2}{\line(1,0){10}}
\end{picture}%
\begin{picture}(30,50)(0,0)
\put(0,-15){\framebox(30,30){$H^{\otimes n}$}}
\put(0,40){\line(1,0){30}}
\end{picture}%
\begin{picture}(16,50)(0,0)
\put(0,0){\line(1,0){16}}
\put(6,-7){\line(1,3){5}}
\put(0,40){\line(1,0){16}}
\end{picture}%
\begin{picture}(20,50)(0,0)
\put(5,20){\makebox(0,0)[l]{$\left. \rule{0mm}{1.0cm} \right\}$}}
\put(13,20){\makebox(0,0)[l]{{ measure}}}
\end{picture}%
}
\vspace{2mm}
\caption{Quantum circuit for the {\bf Sampling Subroutine}.\label{fig:sampling}}
\end{figure}
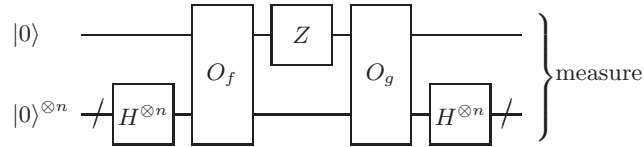

\begin{proof}
  The algorithm relies on the {\bf Sampling Subroutine} described in
  Fig.~\ref{fig:sampling}, where $H$ denotes the standard Hadamard
  gate, $Z$ is a phase gate acting on one qubit as
  $Z:\ket{b}\mapsto (-1)^b\ket{b}$, and $O_f$ is the oracle for
  $f$ acting on $n+1$ qubits as $O_f:\ket{b}\ket{x}\mapsto\ket{b\oplus
    f(x)}\ket{x}$ (similarly for $O_g$). The algorithm works as
  follows:

\algorithm{ {\bf Quantum algorithm}
\begin{enumerate}
 \item Set $i=1$
 \item\label{step:sampling} Run the {\bf Sampling Subroutine}.  Denote
   by $(b_i,u_i)$ the output of the measurement.
\item If $\Span\{u_k|k\in{[i]}\}\neq\ZZ_2^n$, increment $i\rightarrow
  i{+}1$ and go back to step~\ref{step:sampling}. Otherwise set $t=i$
  and continue.
\item\label{step:system} Output ``$s$'', where $s$ is the unique
  solution of
\begin{align*}
 \begin{cases}
 \ip{u_1}s=b_1;\\
 \quad\dots\\
 \ip{u_t}s=b_t.
 \end{cases}
\end{align*}
\end{enumerate}
}

Obviously, this algorithm makes $O(t)$ quantum queries to the oracles
and its complexity is polynomial in $t+n$. The quantum state before the measurement is
\begin{align}
  \ket{0}\ket{0}^{\otimes n}
&\stackrel{H^{\otimes n}}{\longmapsto}\frac{1}{\sqrt{2^n}}\sum_x\ket{0}\ket{x}
\stackrel{O_f}{\longmapsto} \frac{1}{\sqrt{2^n}}\sum_x\ket{f(x)}\ket{x}
\stackrel{Z}{\longmapsto}\frac{1}{\sqrt{2^n}}\sum_x(-1)^{f(x)}\ket{f(x)}\ket{x}\nonumber\\
&\stackrel{O_g}{\longmapsto}\frac{1}{\sqrt{2^n}}\sum_x(-1)^{f(x)}\ket{f(x)\oplus g(x)}\ket{x}\nonumber\\
&\phantom{\longmapsto}=\frac{1}{\sqrt{2^n}}\ket{0}\sum_x\frac{F(x)+F(x+s)}{2}\ket{x}+\frac{1}{\sqrt{2^n}}\ket{1}\sum_x\frac{F(x)-F(x+s)}{2}\ket{x}\nonumber\\
&\stackrel{H^{\otimes n}}{\longmapsto}
  \ket{0}\sum_u\frac{1+\chi_u(s)}{2}\wh F(u)\ket{u}
  +\ket{1}\sum_u\frac{1-\chi_u(s)}{2}\wh F(u)\ket{u}.\label{m_D}
\end{align}
Its measurement therefore always returns a pair
$(b_i,u_i)\in\{0,1\}\times\{0,1\}^n$ where $\ip{u_i}s=b_i$. Moreover,
since by construction $\Span\{u_i|i\in{[t]}\}=\ZZ_2^n$, the system of
equations in step~\ref{step:system} accepts a unique solution that can
only be the hidden shift $s$, thus the final answer of our algorithm
is always correct.

We now show that the algorithm terminates in bounded expected time. We
need to prove that repeatedly sampling using the procedure in
step~\ref{step:sampling} yields $n$ linearly independent vectors
$u_i$, therefore spanning $\ZZ_2^n$, after a bounded expected number
of trials $t$. Let $(B,U)$ be a pair of random variables describing the measurement outcomes for the {\bf Sampling Subroutine}, and $\DU$ denote the marginal distribution of $U$. From the right-hand side of~(\ref{m_D}) it is clear that
\begin{align*}
 \DU(u)\equiv\abs{\wh F(u)}^2.
\end{align*}
Note that this distribution does not depend on $g$.

Let $d_i$ be the dimension of $\Span\{u_k|k\in{[i]}\}$. By
construction, we have $d_1=1,d_t=n$ and $d_{i+1}$ equals either $d_i$
or $d_{i}+1$. Let us bound the probability that $d_{i+1}=d_{i}+1$, or,
equivalently, that $u_{i+1}\notin\Span\{u_k|k\in{[i]}\}$. This
probability can only decrease as $d_{i}$ increases, so let us consider
the worst case where $d_{i}=n-1$. In that case, there exists some
$v\in\ZZ_2^n\setminus\{0\}$ such that $\Span\{u_k|k\in{[i]}\}$ is
exactly the subspace orthogonal to $v$. Then, the probability that
$u_{i+1}$ distributed according to $\DU$ does not lie in this subspace
(and hence $d_{i+1}=d_i+1$) is given by
\begin{align*}
 \PR[u\sim\DU]{\ip vu=1}=\sum_{u:\ip vu=1} \abs{\wh F(u)}^2=\gv,
\end{align*}
which follows from Lemma~\ref{lem:influence}.  Therefore, for any $i$,
the probability that $d_{i+1}=d_i+1$ is at least $ \gf=\min_{v}{\gv},
$ and the expected number of trials before it happens is at most
$1/\gf$. Since $d_i$ must be incremented $n$ times, the expected total
number of trials $t$ is at most $n/\gf$.

Using quantum amplitude amplification, we can obtain a quadratic
improvement over this expected running time. Indeed, instead of
repeating the {\bf Sampling Subroutine} $O(1/\gf)$ times until we
obtain a sample $u$ not in the subspace spanned by the previous
samples, we can use quantum amplitude amplification, which achieves
the same goal using only $O(1/\sqrt{\gf})$ applications of the quantum
circuit in the {\bf Sampling Subroutine} (see~\cite[Theorem~3]{BrassardHMT02}). We therefore obtain a quantum algorithm that
solves the problem with success probability 1 and an expected number
of queries $O(n/\sqrt{\gf})$. \hfill $\Box$
\end{proof}

In case a lower bound on $\gf$ is known, we have the following corollary:
\begin{corollary}\label{cor:complexity-promise}
  There exists a quantum algorithm that solves an instance of \BHSP\
  defined over the function $f$, with the promise that
  $\gf\geq\delta$, with success probability at least $1-\varepsilon$
  and using at most $O(n\log(1/\varepsilon)/\sqrt{\delta})$ oracle
  queries.  The algorithm takes expected time polynomial in the number
  of queries.
\end{corollary}

\begin{proof}
  This immediately follows from Markov's inequality, since it implies
  that the algorithm in Theorem~\ref{thm:complexity} will still
  succeed with constant probability even when we stop after a time
  $\Theta(n/\sqrt{\gf})$ if it has not succeeded so far. \hfill $\Box$
\end{proof}

\section{Classical complexity of random instances of \BHSP}

In this section we show that a uniformly chosen instance of \BHSP\ is
exponentially-hard classically with high probability.

\begin{lemma}\label{l_cla_h}
  A classical algorithm solving a uniformly random instance of \BHSP\
  with probability at least $1/2$ makes $\Omega(2^{n/2})$ oracle
  queries.
\end{lemma}

\begin{proof}
  Consider a classical algorithm $\Ac$ that makes $\qc$ queries to the
  oracles $A_f$ and $A_g$ and with probability at least $1/2$ returns
  the unique $s$ satisfying $g(x)\equiv f(x+s)$
  (cf.~Definition~\ref{d_BHSP}).  For notational convenience we assume
  that $\Ac$ only makes duplicated queries $\left(f(x),g(x)\right)$.
  This can at most double the total number of oracle calls.

  Consider the uniform distribution of $f:\ZZ_2^n\to\ZZ_2$ and
  $s\in\ZZ_2^n$, and let an input instance of \BHSP\ be chosen
  accordingly.  Let $\left(X_1,\dots, X_{\qc}\right)$ be random
  variables representing the queries made by $\Ac$.  Then by the
  correctness assumption, the values $f(X_1),g(X_1),\dots,
  f(X_{\qc}),g(X_{\qc})$ can be used to predict $s$ with probability
  at least $1/2$.

  First we observe that if, after $k$ queries, it holds that
  $X_i-X_j\neq s$ for every $i,j\in[k]$, then even conditionally on
  the values of $f(X_1),g(X_1),\dots, f(X_{k}),g(X_{k})$ every
  $s\notin\{X_i-X_j|i,j\in[k]\}$ has exactly the same probability to
  occur. More precisely, if $S_k=\{X_i-X_j|i,j\in[k]\}$ and $E_k$ is
  the event that $s\in S_k$, we have
\begin{align}
 \PR{s=s_0|\neg E_k}=\frac{1}{2^n-\sz{S_k}}
\leq \frac{1}{2^n-k^2}
\end{align}
for any $s_0\notin S_k$ and $0\leq k\leq\qc$.  In other words, modulo
``$s\notin S_k$'' the actual values of $f$ and $g$ at points
$\{X_i|i\in[k]\}$ provide no additional information about $s$, and the
best the algorithm can do in that case is a random guess, which
succeeds with probability at most $1/(2^n-k^2)$.

Now let us analyze the probability that
$S_{\qc}=\{X_i-X_j|i,j\in[\qc]\}$ contains $s$, that is,
$\PR{E_{\qc}}$.  Since $\sz{S_{k+1}}-\sz{S_k}\leq k$,
we have by the union bound
\begin{align*}
 \PR{E_{k+1}|\neg E_k}
&\leq\sum_{s_0\in S_{k+1}}\PR{s=s_0|\neg E_k} \leq \frac{k}{2^n-k^2}.
\end{align*}
Consequently,
\begin{align*}
  \PR{E_{\qc}}\leq \frac{\sum_{k=0}^{\qc-1}k}{2^n-\qc^2}
  \leq\frac{\qc^2}{2^n-\qc^2}.
\end{align*}
Finally, we can bound the probability that the algorithm succeeds
after $\qc$ oracle queries as
\begin{align*}
 \PR{\Ac\textrm{ succeeds}}
&=\PR{\Ac\textrm{ succeeds}|E_{\qc}}\cdot\PR{E_{\qc}}\\
&\qquad +\PR{\Ac\textrm{ succeeds}|\neg E_{\qc}}\cdot\PR{\neg E_{\qc}}\\
&\leq \PR{E_{\qc}}+\PR{\Ac\textrm{ succeeds}|\neg E_{\qc}}\leq\frac{\qc^2+1}{2^n-\qc^2},
\end{align*}
which is larger than $1/2$ only if $\qc\in\Omega\left(2^{n/2}\right)$,
as required. \hfill $\Box$
\end{proof}

We are now ready to state our main theorem which is an exponential
quantum-classical separation for an average case problem.

\begin{theorem}[Average case exponential separation]
Let $({\cal O}_f,
  {\cal O}_g)$ be an instance of a Boolean hidden shift problem (BHSP)
  where $g(x)=f(x+v)$ and $f$ and $v$ are chosen uniformly at random.
  Then there exists a quantum algorithm which finds $v$ with bounded error using $O(n)$
  queries and in $O(\mathrm{poly}(n))$ time whereas any classical algorithm needs
  $\Omega(2^{n/2})$ queries to achieve the same task.
\end{theorem}
\begin{proof}
  For a fixed $v$ and randomly chosen $f$, consider the $2^{n-1}$
  mutually independent events ``$f(x)=f(x+v)$''. By definition of
  $\gamma_{f,v}$ and the Chernoff bound, the probability that
  $\gamma_{f,v}<1/3$ is at most $e^{-\Omega(2^n)}$. Since this is
  double-exponentially small in $n$ we obtain from an application of
  the union bound to the $2^n$ possible values of $v$ that if
  $f:\ZZ_2^n\to\ZZ_2$ is chosen uniformly at random then ${\bf
    Pr}_f[\gf<1/3]\in e^{-\Omega(2^n)}$. We now apply Corollary
  \ref{cor:complexity-promise} for constant $\gf$ to obtain a quantum
  algorithm that uses at most $O(n)$ queries and outputs the correct
  hidden shift $v$ with constant probability of success (i.e.,
  $\varepsilon$ is chosen to be constant). Combining this with the
  exponential lower bound from Lemma~\ref{l_cla_h} implies that there
  is an exponential gap between the classical and quantum complexity
  of the \BHSP\ defined over a random Boolean function. \hfill $\Box$
\end{proof}

\section{Discussion and open problems}
We presented a quantum algorithm for the Boolean hidden shift problem
that is based on sampling from the space of vectors that are
orthogonal to the hidden shift. It should be noted that our algorithm
reduces to one of the two algorithms given in \cite{Roetteler:2010} in
case the function is a bent function. We related the running time and
the query complexity of the algorithm to the minimum influence of the 
function and showed that for random functions these complexities are polynomial. This leads to
an average case exponential separation between the classical and quantum time complexity
for Boolean functions. An interesting question is whether these methods 
can be generalized and adapted for the case of non-Boolean functions
also. Furthermore, we conjecture that the complexity of our quantum algorithm is optimal up to polynomial factors for any function.

\subsection*{Acknowledgments}

The authors acknowledge support by ARO/NSA under grant
W911NF-09-1-0569.  We wish to thank Andrew Childs, Sean Hallgren,
Guosen Yue and Ronald de Wolf for fruitful discussions.

\bibliography{hidden-shift}

\begin{appendix}

\section{Proof of Lemma~\ref{lem:influence}\label{app:influence}}
\begin{lemma2}\label{lem2:influence}
$\gv=\sum_{u:\ip vu=1} \abs{\wh F(u)}^2.$
\end{lemma2}
\begin{proof}
Let us consider the following function
$
 \tf(x)\deq F(x)-F(x+v).
$
Its Fourier transform reads
\begin{align*}
 \wh\tf(u)
 =\E[x]{F(x)\cdot\chi_u(x)-F(x+v)\cdot\chi_u(x)}
 =(1-\chi_u(v))\cdot\wh F(u).
\end{align*}
Therefore, we have
\begin{align*}
\sum_{u:\ip vu=1} \abs{\wh F(u)}^2
 &=\frac{1}{4}\sum_{u\in\ZZ_2^n}\abs{(1-\chi_u(v))\cdot\wh F(u)}^2
 =\frac{1}{4}\sum_{u\in\ZZ_2^n}\abs{\wh\tf(u)}^2\\
 &=\frac{1}{4}\E[x]{\abs{\tf(x)}^2}=\PR[x]{F(x)\neq F(x+v)}=\gv,
\end{align*}
where in the second line we have used Parseval's identity.
\hfill $\Box$
\end{proof}

\MyComment{Look for ...-s}

\MyComment{Spell-check}

\mynb{To consider next:
 1) What if f has nontrivial "self-shifts"?
 2) Non-Boolean case
 3) Partial answers due to "unfriendly" structure of f ('1' above is an extreme case)
}


\end{appendix}

\end{document}